\documentclass[journal]{IEEEtran}
\IEEEoverridecommandlockouts
\usepackage[utf8]{inputenc}
\usepackage{amsmath,amsfonts,amsthm}
\usepackage{algorithm,algorithmic}
\usepackage{booktabs}
\usepackage{bm}
\usepackage{comment}
\usepackage{chngpage}
\usepackage{enumerate}
\usepackage{paralist}
\usepackage{extramarks}
\usepackage{fancyhdr}
\usepackage{graphicx,float,wrapfig}
\usepackage{lastpage}
\usepackage{setspace}
\usepackage{soul,color}
\usepackage{thmtools}
\usepackage{thm-restate}
\usepackage{xspace}
\usepackage{indentfirst}

\usepackage{nohyperref}
\usepackage{cleveref}

\makeatletter
\DeclareRobustCommand\onedot{\futurelet\@let@token\@onedot}
\def\@onedot{\ifx\@let@token.\else.\null\fi\xspace}

\def\eg{\emph{e.g}\onedot} 
\def\ie{\emph{i.e}\onedot} \def\Ie{\emph{I.e}\onedot}
\def\cf{\emph{cf}\onedot}

\def\st{$st~$}
\makeatother

\theoremstyle{definition}
\newtheorem{defn}{Definition}

\theoremstyle{plain}
\newtheorem{thm}{Theorem}
\newtheorem{lem}{Lemma}

\newtheorem{crl}{Corollary}

\theoremstyle{remark}
\newtheorem*{rmk}{Remark}

\newtheoremstyle{problem}
  {}
  {}
  {}
  {0pt}
  {\bfseries}
  {:}
  { }
  {\thmname{#1}\thmnumber{ #2}\thmnote{ (#3)}}

\newcommand{\ssl}{\mathcal{S}}

\newif\iftechrep \techrepfalse

\catcode`T=12 \catcode`R=12	
\def\checkiftechreport#1{
\expandafter\iistechreport#1TR. \techreptrue\fi}
\def\iistechreport#1TR#2.{\def\tmp{#2}\ifx\tmp\empty\else}
\catcode`T=11 \catcode`R=11
\def\checkTR{\checkiftechreport{\jobname}}
\checkTR

\title{Robust Routing in Interdependent Networks}
\author{Jianan Zhang, and Eytan Modiano
\thanks{Part of the material in this paper was presented at IEEE International Conference on Computer Communications (INFOCOM), 2017.

The authors are with the Laboratory for Information and Decision Systems, Massachusetts Institute of Technology. This work was supported in part by DTRA grants HDTRA1-13-1-0021 and HDTRA1-14-1-0058.}}
\begin{document}
\maketitle
\begin{abstract}
We consider a model of two interdependent networks, where every node in one network depends on one or more supply nodes in the other network and a node fails if it loses all of its supply nodes.
We develop algorithms to compute the failure probability of a path, and obtain the most reliable path between a pair of nodes in a network, under the condition that each supply node fails independently with a given probability. Our work generalizes the classical shared risk group model, by considering multiple risks associated with a node and letting a node fail if all the risks occur. 
Moreover, we study the diverse routing problem by considering two paths between a pair of nodes.
We define two paths to be $d$-failure resilient if at least one path survives after removing $d$ or fewer supply nodes, which generalizes the concept of disjoint paths in a single network, and risk-disjoint paths in a classical shared risk group model. We compute the probability that both paths fail, and develop algorithms to compute the most reliable pair of paths. 
\end{abstract}

\section{Introduction}

Many modern systems are interdependent, such as smart power grids, smart transportation, and other cyber-physical systems \cite{yagan2012optimal, rosato2008modelling, gu2011onset, parandehgheibi2013robustness, parandehgheibi2015modeling}. In interdependent networks, one network depends on another to properly operate. For example, in smart grids, power generators rely on messages from the control center to adjust to the power demand fluctuations, while the control center relies on the electric power to operate. Due to the interdependence, failures in one network may cascade to another. It is important to understand the robustness of interdependent networks which are prone to cascading failures.

Most previous studies on interdependent networks have focused on the network connectivity based on random graph models, in the asymptotic regime where the number of nodes approaches infinity \cite{buldyrev2010catastrophic, shao2011cascade}. The finite-size arbitrary-topology graph models, which represent real-world communication and physical networks, have been largely overlooked in the interdependent networks literature. A few exceptions include \cite{parandehgheibi2013robustness, parandehgheibi2015modeling}, which model interdependent power grids and communication networks by graphs with topologies specified by the real networks. In \cite{zhang2018connectivity}, we model interdependent networks by graphs with specified topologies, which can be tailored for a wide range of applications, and study the connectivity in interdependent networks. 

In this paper, we study robust routing problems in interdependent networks, by characterizing the effects of node failures in one network on the nodes and pathes in the other network. For an overview of the problems and challenges, it is helpful to consider a simplified scenario where a demand network depends on a supply network, illustrated by Fig. \ref{fig:interdep}. Every node in the demand network is supported by one or more nodes in the supply network. Thus, nodes in the demand network and nodes in the supply network can be viewed as demand nodes and supply nodes, respectively. Given that a demand node fails if it loses all of its supply nodes, supply node failures may lead to correlated demand node failures, which makes it difficult to route traffic through reliable paths in the demand network. We develop techniques to tackle the failure correlation. This simplified one-way dependence exists in real-world systems. For example, routers and processors in a communication network depend on the electric power. Moreover, as we will see later, the analysis based on this simplified scenario can be applied to interdependent networks under certain assumptions.
\begin{figure}[h]
\begin{centering}
\includegraphics[width=0.85\linewidth]{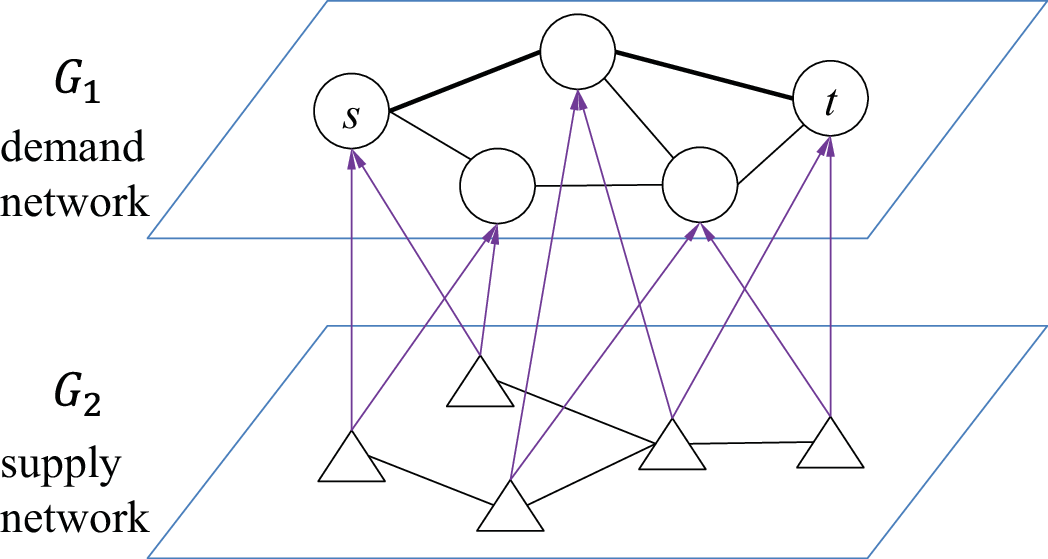}
\caption{Every node in the demand network $G_1$ is supported by two nodes in the supply network $G_2$.}
\label{fig:interdep}
\end{centering}
\end{figure}

The robust routing problems have been extensively studied under both independent failure and correlated failure scenarios. If edges or nodes fail independently, the most reliable path between a source-destination pair can be viewed as a shortest path, where the length is a function of the failure probability. In the case of correlated failures, it is difficult to find a path with any performance guarantee in general \cite{yang2015shortest}. If correlation only exists among edges or nodes that fail simultaneously, the network can be viewed using a shared risk group model (Fig.~\ref{fig:srg}) \cite{xu2004failure, coudert2007shared}. The shared risk group model captures correlated failures in an overlay network when underlay failures occur, and is commonly used to study the cross-layer reliability, such as logical link failures caused by fiber failures in optical networks \cite{hu2003diverse, yuan2005minimum, lee2010diverse, habib2013disaster}. The most reliable path contains the smallest number of risks if all risks are equally likely to occur, and can be obtained by integer programming~\cite{hu2003diverse}.

\begin{figure}[h]
\begin{centering}
\includegraphics[width=0.77\linewidth]{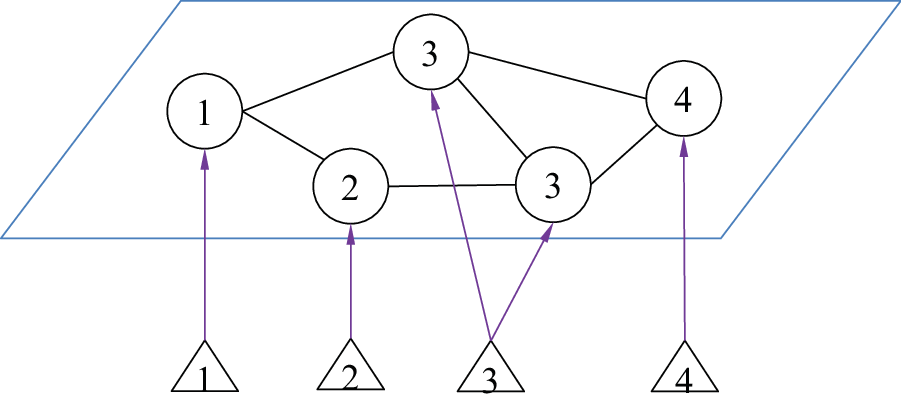} 
\caption{A shared risk node group model, where nodes labeled by the same number share the same risk.}
\label{fig:srg}
\end{centering}
\end{figure}

Interdependent networks have similarities with the classical shared risk group model, in that two demand nodes share a risk if they have at least one common supply node. However, the key difference is that a demand node does not necessarily fail if a risk occurs (\ie, a supply node fails), since a demand node may have multiple supply nodes, whereas a node fails if its associated risk occurs in the classical shared risk group model.

The reliability of a path in interdependent networks, in contrast to the classical shared risk group model, can no longer be characterized by the number of risks that the path contains. For example, if all the nodes in a path depend on a single supply node and thus the path has a single risk, removing a single supply node would disconnect the path. In contrast, if every node in a path has multiple supply nodes, the path would be more robust and can resist a larger number of supply node failures, although the path has more ``risks''.

In addition to the most reliable path, a backup path can be used to further improve reliability, through \emph{diverse routing}. 
Intuitively, a pair of reliable paths should share the minimum number of risks (or be risk-disjoint) in the shared risk group model \cite{hu2003diverse, lee2010diverse}. However, in interdependent networks, it is easy to construct examples where two paths that share many supply nodes can withstand a larger number of supply node failures than two paths that share a smaller number of supply nodes (\eg, Fig. \ref{fig:twopaths}). New metrics, other than the number of risks shared by two paths, need to be identified to characterize their reliability.
\begin{figure}[h]
\begin{centering}
\includegraphics[width=\linewidth]{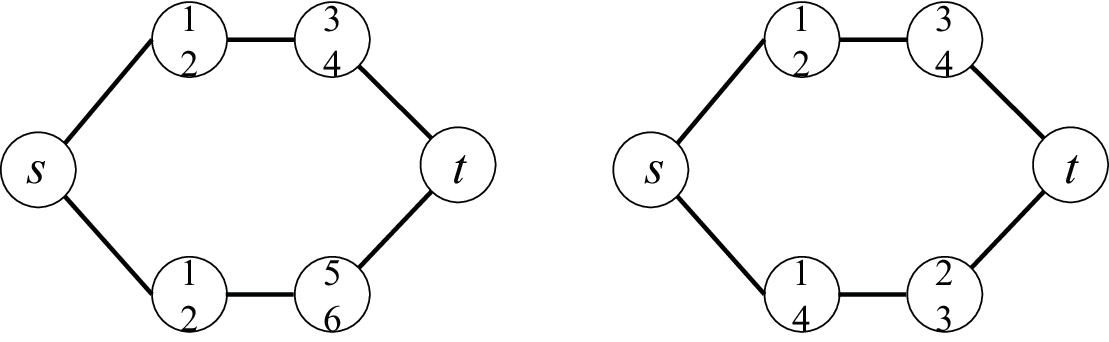}
\caption{The two numbers in each node represent its two supply nodes. The two \st paths in the left figure share two supply nodes, and can be both disconnected after removing supply nodes $\{1,2\}$. The two \st paths in the right figure share four supply nodes, but they cannot be both disconnected after removing any two supply nodes.}
\label{fig:twopaths}
\end{centering}
\end{figure}

Diverse routing problems have been studied under correlated link failures. The correlation between a pair of logical links is obtained either by measurement \cite{cui2002backup} or by analysis of the underlay physical topology \cite{zheng2014cross}. Heuristic algorithms have been developed to find multiple reliable paths, and their performance was evaluated by simulation \cite{cui2002backup, tang2005improving, Fei2006}. In contrast, we explicitly bound the gap between the failure probability and the optimization objective, and develop algorithms that have provable performance.

In this paper, we develop an analytically tractable framework to study the following robust routing problems in interdependent networks.

\texttt{Single-path routing:}
Compute the probability that a specified path fails. Obtain the most reliable path between a source-destination pair.


\texttt{Diverse routing:}
Compute the probability that two specified paths both fail. Obtain the pair of most reliable paths between a source-destination pair.

By generalizing the concept of disjoint paths to interdependent networks, we characterize the level of disjointness between two paths to study diverse routing.
In contrast to the classical shared risk group model where a node fails if its risk occurs, in interdependent networks a node fails if a combination of risks occur. In view of this, our methods extend the shared risk group model. To the best of our knowledge, this paper is the first to study the robust routing problem in interdependent networks. A preliminary version of this paper appeared in \cite{zhang2017robust}.

The rest of the paper is organized as follows. In Section \ref{sc:model}, we state our model for interdependent networks and failures. In Section \ref{sc:single}, we prove the complexity, and develop approximation algorithms to compute the path failure probability. 
In Section \ref{sc:search}, we develop algorithms to find the most reliable path between a pair of nodes. 
In Section \ref{sc:pair}, we study the diverse routing problem in interdependent networks, and find a pair of reliable paths whose failure probability is minimized. Section \ref{sc:numerical} provides numerical results.
Finally, Section \ref{sc:conclusion} concludes the paper.

\section{Model}
\label{sc:model}
We consider a demand network $G_1$ and a supply network $G_2$, where every demand node in $G_1$ depends on one or more supply nodes in $G_2$. We assume that every supply node provides substitutional supply to the demand nodes, and a demand node is functioning if it is directly connected to at least one supply node. To study the impact of node failures in $G_2$ on $G_1$, it is equivalent to study the following model.

Consider a graph $G(V,E,\mathcal S_V)$, where nodes $V$ and edges $E$ are identical to nodes and edges in $G_1$, and $\mathcal S_V$ are the supply node sets, each of which is a set of nodes in $G_2$ that provide supply to a node in $V$. In this model, each node $v_i \in V$ is a demand node, supported by a set of supply nodes $S_i \in \mathcal S_V$, and $v_i$ fails if all the nodes in $S_i$ fail. (Note that nodes $V$ may have different number of supply nodes.) Finally, let $s,t\in V$ be a source-destination pair.

Under the condition that supply nodes fail independently with given probabilities, and following the convention that $s,t$ do not fail, we study the robust routing problems in $G(V, E, \mathcal S_V)$.

\begin{rmk}
 The analysis for this model can be directly applied to interdependent networks, as long as the interdependence is bidirectional (\ie, if $v \in G_1$ depends on $u \in G_2$, then $u$ depends on $v$ as well) and failures initially occur in one network. It suffices to observe that, given a set of failed nodes $S \subseteq G_2$, a node $v \in G_1$ fails if and only if its supply nodes are all in $S$. Notice that the failure of $v$ does not further lead to node failures in $G_2$, because all the nodes that $v$ supports, which are exactly the supply nodes for $v$ due to the bidirectional interdependence, have failed.
\end{rmk}

\section{Computing the reliability of a path}
\label{sc:single}
If every node has a single supply node, the path failure probability is given by $1 - (1 - p)^r$, where each supply node fails independently with probability $p$ and the path is supported by $r$ supply nodes. In contrast, if every node has more than one supply node, computing the path failure probability becomes $\#P$-hard. The proof can be found in the Appendix.
\begin{restatable}{thm}{thmSharpP}
Computing the failure probability of a path is $\# P$-hard, if every node has two or more supply nodes and each supply node fails independently with probability $p$.
\end{restatable}

Although it is $\#P$-hard to compute, the path failure probability can be well approximated. We apply the solution to the \emph{DNF probability} problem and propose an $(\epsilon, \delta)$-approximation algorithm based on importance sampling, which approximates the path failure probability to within a multiplicative factor $1 \pm \epsilon$ with probability at least $1 - \delta$. 

The DNF probability problem computes the probability that a Disjunctive Normal Form (DNF) formula is true, when literals are set to be true independently with given probabilities. A DNF formula is a disjunction of clauses, each of which is a conjunction of literals, and takes the following form: $(x_1^1 \land \dots \land x_{n_1}^1)\lor (x_1^2 \land  \dots \land x_{n_2}^2) \lor \dots \lor (x_1^m \land \dots \land  x_{n_m}^m)$. Let $v_1 - \dots - v_m$ be a path in $G(V,E,\ssl_V)$.
The key observation is that computing the path failure probability can be formulated by a DNF probability problem, in which a clause $C_i$ represents a node $v_i$ in the path and the literals $x_j^i$ in clause $C_i$ represent the supply nodes of $v_i$. 
For completeness, we state Algorithm \ref{al:sampling} that approximates the path failure probability, by adapting the algorithm that approximates the DNF probability in \cite{karp1989monte}.


\begin{algorithm}[h]
\caption{Estimating the path failure probability based on importance sampling.}
\label{al:sampling}
\textbf{Initialization:}
\begin{enumerate}
\item
 Given a path $\{ v_1,  v_2, \dots,  v_{ m} \}$, let $\{ u_j^i | j = 1,\dots, n_s( v_i) \}$ denote the set of supply nodes of $ v_i$, where $n_s( v_i)$ is the number of supply nodes of $ v_i$.
\end{enumerate}

\textbf{Main loop:}
\begin{enumerate}
\setcounter{enumi}{1}
\item Among $\{ v_1,  v_2, \dots,  v_{ m} \}$, randomly choose $ v_i$ with probability $$\prod_{1 \leq j \leq n_s( v_i)} p(u_j^i) / \sum_{1 \leq k \leq  m} \prod_{1 \leq j \leq n_s( v_k)} p(u_j^k).$$ If every demand node has an identical number of supply nodes, and the supply node failure probability $p(u_j^i)$ is identical, then node $ v_i$ is chosen with probability $1 /  m$.
\item If $ v_i$ is chosen, set all of its supply nodes $\{ u_j^i | j = 1,\dots, n_s( v_i) \}$ to be failed. The other supply nodes are randomly set to be failed with their respective failure probabilities. Let $U$ denote the set of failed supply nodes.
\item Test whether $ v_i$ is the first failed node among $\{ v_1,  v_2, \dots,  v_{ m} \}$, given that $U$ fail (and no other supply nodes fail). If true, set $I = 1$; otherwise, set $I = 0$. Repeat the loop for $a = 3 m \ln(2/\delta) / \epsilon^2$ iterations.
\end{enumerate}

\textbf{Result:}
\begin{enumerate}
\setcounter{enumi}{4}
\item Count the number of $I = 1$ and denote the number by $b$. An $(\epsilon, \delta)$-approximation of the path failure probability is given by $b/a \sum_{1 \leq k \leq m} \prod_{1 \leq j \leq n_s(v_k)} p(u_j^k)$.
\end{enumerate}
\end{algorithm}

The intuition behind this importance sampling algorithm is as follows. Some events, although rare, are important in determining the path failure probability, especially when the path failure probability is small. The algorithm samples in a space consisting of important events, each of which is a set of supply node failures $U$ that lead to the path failure. In this space, the failure of $U$ may appear multiple times, given that multiple choices of $v_i$ in Step 2 may lead to the same $U$ in Step 3. The algorithm then remove the duplicated $U$ via sampling in Step~4.


To prove the correctness of the algorithm, we take the following two steps. First, following a similar analysis to \cite{karp1989monte}, we prove that the path failure probability is given by $E[I] \sum_{1 \leq k \leq m} \prod_{1 \leq j \leq n_s(v_k)} p(u_j^k)$, where $E[I]$ is the expectation of $I$ in Step 4 of the algorithm. Second, by repeating the loop a sufficiently large number of times, $E[I]$ can be approximated to within factor $1 \pm \epsilon$ with probability at least $1 - \delta$. 
The details of the proof can be found in the Appendix.

The advantage of this algorithm over a na\"{\i}ve Monte-Carlo algorithm (\eg, by repeatedly simulating the supply node failure events and counting the fraction of trials in which the path fails) is that the number of iterations in the na\"{\i}ve Monte-Carlo algorithm is large when the path failure probability is small\footnote{If $F$ occurs in $b$ out of $a$ trials, $\Pr(F) \in (1 \pm \epsilon) b/a$ with probability $1 - \delta$, under the condition that $b=\Omega(\ln(1/\delta) / \epsilon^2)$. The total number of trials $a = \Omega(\ln(1/\delta) / \epsilon^2)/\Pr(F)$ is large when $\Pr(F)$ is small.}. In contrast, by sampling in a more important space, the number of iterations is reduced. Note that the only quantity that needs to be estimated in Algorithm \ref{al:sampling} by simulation is $E[I]$, and that $\Pr(I = 1) \geq 1/m$. We conclude this section by the following theorem, whose proof is in the Appendix.
\begin{restatable}{thm}{thmtime} \label{th:time}
  The path failure probability can be estimated to within a multiplicative factor $1 \pm \epsilon$ with probability $1 - \delta$, in time $O(m^2 n_s \ln(1/\delta) / \epsilon^2)$, where $m$ is the path length and $n_s$ is the maximum number of supply nodes for a demand node.
\end{restatable}

Although the failure probability of a specific path can be well approximated by the importance sampling algorithm, the algorithm hardly gives an intuition for path properties that characterize a reliable path. In the remainder of this section, we develop indicators and bounds on the path failure probability, which can be used for finding the most reliable path.

\subsection{Small and identical failure probability}
Consider a path $v_1 - \dots - v_m$ in $G(V,E,\ssl_V)$. Let $F_i$ denote the event that all the supply nodes of $v_i$ fail. Let $F$ denote the event that the path fails. Clearly, the path fails if at least one node $v_i$ loses all of its supply nodes ($F = \cup_{1 \leq i \leq m} F_i$). 

By the inclusion-exclusion principle, we have
\begin{align}
\Pr (F) &= \sum_{1 \leq i \leq m} \Pr(F_i) - \sum_{1 \leq i_1 < i_2 \leq m} \Pr(F_{i_1} \cap F_{i_2}) \nonumber \\
& + \dots + (-1)^{m-1} \Pr(F_1 \cap F_2 \dots \cap F_m).
\label{eq:inc}
\end{align}
Directly computing the path failure probability is difficult, given that there are ${m \choose j}$ summations in the $j$-th term of the inclusion-exclusion formula. We first reduce the number of events in the inclusion-exclusion formula, and then further simplify the computation under the condition that the supply node failure probability is small and identical.

To reduce the number of events, some \emph{redundant} events can be ignored. For example, if $F_i$ occurs only if $F_j$ occurs, then the event $F_i$ is redundant in determining $F$ with the knowledge of $F_j$. To see this, note that 1) if $F_j$ occurs, then the path fails regardless of $F_i$; 2) if $F_j$ does not occur, then $F_i$ does not occur as well. If the supply nodes of $v_j$ form a subset of the supply nodes of $v_i$, then $F_i$ is redundant. With an abuse of language, we call a node $v_i$ redundant if $F_i$ (\ie, the state of $v_i$) is redundant. With this simplification, we derive the following result.


Let $n_s(v_i)$ denote the number of distinct supply nodes of $v_i$. Let $n_s^{\min} = \min_{1 \leq i \leq m} n_s(v_i)$. After removing the redundant nodes sequentially, let $\bar m$ be the number of remaining nodes that each have $n_s^{\min}$ supply nodes. The path failure probability can be estimated by the following theorem.
\begin{thm}
\label{th:singleGeneral}
  If every supply node fails independently with probability $p \leq \epsilon/m$, then the path failure probability satisfies $(1 - \epsilon) \bar m p^{n_s^{\min}} \leq \Pr(F) \leq (1 + \epsilon) \bar m p^{n_s^{\min}}$.  
\end{thm}
\begin{proof}
  We first reduce the number of failure events that appear in the inclusion-exclusion formula by removing the redundant nodes. Note that determining whether a node is redundant and removing the redundant node are done sequentially. Thus, among the set of nodes that have the same supply nodes, one node remains. Let $D$ denote the nodes in the path excluding the redundant nodes.

  First, we consider the first term in Eq. (\ref{eq:inc}) that provides an upper bound on the path failure probability, known as the union bound. Let $D_1 \subset D$ denote the set of nodes that each have $n_s^{\min}$ supply nodes, and let $\bar m = |D_1|$. The remaining nodes $D_2 = D \setminus D_1$ each have $n_s^{\min}+1$ or more supply nodes. Thus, the first term of Eq. (\ref{eq:inc}) is at most
  \begin{align*}
    \Pr(F) \leq & \bar m p^{n_s^{\min}} + (m - \bar m) p p^{n_s^{\min}} \\
    \leq & \bar m p^{n_s^{\min}} + \epsilon p^{n_s^{\min}},
  \end{align*}
  for $p \leq \epsilon/m$.

  Next, we consider the first two terms that provide a lower bound on the path failure probability (\cf Bonferroni inequalities). For any pair of nodes $v_j, v_k \in D$, the union of their supply node sets contains at least $\max(n_s(v_j), n_s(v_k)) + 1$ nodes, because neither supply node set includes the other as a subset. At least $n_s^{\min} + 1$ supply nodes have to be removed in order for a pair of nodes in $D_1$ to fail. At least $n_s^{\min} + 2$ supply nodes need to be removed in order for a pair of nodes to fail if at least one node belongs to $D_2$. The absolute value of the second term is at most ${\bar m \choose 2} p^{n_s^{\min} + 1} + [{m \choose 2} -  {\bar m \choose 2}] p^{n_s^{\min} + 2}$. A lower bound on $\Pr(F)$ is
  \begin{align*}
    \Pr(F) \geq & \bar m p^{n_s^{\min}} - \Big(\frac{\bar m^2}{2} p p^{n_s^{\min}} + \frac{m^2}{2} p^2 p^{n_s^{\min}} \Big) \\
    \geq & \bar m p^{n_s^{\min}} - \epsilon \bar m p^{n_s^{\min}},
  \end{align*}
  for $p \leq \epsilon/m$.
\end{proof}

For the special case where every node in the path has the same number $n_s$ of distinct supply nodes, let $\bar m$ be the number of nodes, in the path, among which no pair of nodes share the same set of $n_s$ supply nodes. The following stronger result can be proved in a similar approach.
\begin{crl}
\label{th:single}
If every supply node fails independently with probability $p \leq 2 \epsilon /\bar m$, then the path failure probability satisfies $(1 - \epsilon)\bar m p^{n_s} \leq \Pr(F) \leq \bar m p^{n_s}$.
\end{crl}
\begin{proof}
 If node $v_i$ and $v_j$ in the path share the same set of supply nodes, then $v_i$ and $v_j$ must fail simultaneously, and $F = \cup_{1 \leq i \leq m} F_i = \cup_{1 \leq i \leq m, i \neq j} F_i$. Thus, in the calculation of path failure probability $\Pr(F)$, nodes that have the same set of supply nodes can be represented by a single node.

 Let $\{\bar v_1, \bar v_2, \dots, \bar v_{\bar m}\}$ denote the nodes in the path such that the supply nodes of $\bar v_i$ differ from the supply nodes of $\bar v_j$ by at least one supply node, $i \neq j$, $\bar m \leq m$. The first term of Eq. (\ref{eq:inc}) is $\bar m p^{n_s}$, since every node has $n_s$ distinct supply nodes and the probability that a node fails is $ p^{n_s}$. Moreover, the union of the supply nodes of $\bar v_i$ and $\bar v_j$ has size at least $n_s + 1$, and the probability that both $\bar v_i$ and $\bar v_j$ fail (because of their supply nodes' failures) is at most $p^{n_s + 1}$. The absolute value of the second term of Eq. (\ref{eq:inc}) is at most ${ \bar m \choose 2 } p^{n_s + 1} \leq \bar m^2 p^{n_s + 1} / 2 \leq \epsilon \bar m p^{n_s}$, for $p \leq 2 \epsilon / \bar m$. Therefore, $\Pr(F) \in [(1 - \epsilon) \bar m p^{n_s}, \bar m p^{n_s}]$ given that $p \leq 2 \epsilon/\bar m$.
\end{proof}

Thus, we have obtained the following two \emph{reliability indicators} for a path. These combinatorial properties are useful in finding a reliable path, which will be studied in the next section.
\begin{itemize}
  \item $n_s^{\min}$: the minimum number of distinct supply nodes for a node in the path.
  \item $\bar m$: the number of combinations of $n_s^{\min}$ supply node failures that lead to the failure of at least one node in the path.
\end{itemize}

\subsection{Arbitrary failure probability}
\label{sc:oneDNF}
In contrast with the case where supply node failure probability is small and identical, it is difficult to characterize the reliability of a path by its combinatorial properties, with limited knowledge of node failure probabilities. Therefore, we obtain bounds on path failure probability that will be useful in finding a reliable path.

First, we develop an upper bound on the path failure probability. Let $p(v_i)$ be the failure probability of node $v_i$, under the condition that each of its supply nodes $u_j^i$ fails independently with probability $p(u_j^i)$. The path failure probability, under the condition that the failures of $V$ are positively correlated, is no larger than the path failure probability by assuming that the failures of $V$ are independent.
\begin{lem}
\label{th:upper}
   The failure probability of a path $P$ where a supply node $u_j^i$ fails independently with probability $p(u_j^i)$ is upper bounded by $1-\prod_{v_i \in P}(1-p(v_i))$.
\end{lem}
\begin{proof}
  If nodes $v_i$ and $\cup_{k} v_k$ do not share any supply node, then the event that $v_i$ survives and the event that $\cup_{k} v_k$ survive are independent. Otherwise, if they share one or more common supply nodes, the two events are positively correlated. Therefore,
  $$\Pr(v_i \text{~and } \cup_k v_k \text{~survive}) \geq \Pr(v_i \text{~survives}) \Pr(\cup_k v_k \text{~survive}),$$
  and
  $$\Pr(v_i \text{~survives}| \cup_k v_k \text{~survive}) \geq 1 - p(v_i).$$

  The reliability of a path $P=v_1-v_2-\dots-v_m$ is given by
  \begin{align*}
    \Pr(P \text{~survives}) & = \Pr(\cup_{k \in \{1,\dots,m\}}v_k \text{~survive}) \\
    \hspace{-4mm} & \hspace{-4mm}= \Pr(v_1 \text{~survives}) \Pr(v_2 \text{~survives}|v_1 \text{~survives}) \\
    & \hspace{-4mm} \dots \Pr(v_m \text{~survives}|\cup_{k \in \{1,\dots,m-1\}}v_k \text{~survive}) \\
    & \hspace{-4mm}\geq  \prod_{v_i \in P} (1 - p(v_i)).
  \end{align*}
\end{proof}

Then, we develop a lower bound on the path failure probability. The intuition is as follows. After replacing a supply node that supports multiple demand nodes by multiple independent supply nodes with sufficiently small failure probability, the path failure probability does not increase. 
In the original graph $G(V,E,\ssl_V)$, consider a node $v_i \in V$. Let $U^i$ denote the set of supply nodes of $v_i$, let $u_j^i \in U^i$ denote one supply node, let $p(u_j^i)$ denote the failure probability of $u_j^i$, and let $n_d(u_j^i)$ denote the number of nodes that $u_j^i$ supports. Let $\tilde p(v_i) = \prod_{u_j^i \in U^i} \tilde p(u_j^i)$ denote the failure probability of $v_i$ if $u_j^i$ fails independently with probability $\tilde p(u_j^i) = 1 - (1 - p(u_j^i))^{1/n_d(u_j^i)}$. A lower bound on the path failure probability is as follows, whose proof follows a similar technique in \cite{gatterbauer2014oblivious} and is in the technical report.
\begin{lem}
\label{th:lower}
  The failure probability of a path $P$ where a supply node $u_j^i$ fails independently with probability $p(u_j^i)$ is lower bounded by $1-\prod_{v_i \in P}(1-\tilde p(v_i))$.
\end{lem}
\begin{proof}
  Let $U^P = \cup_{1 \leq i \leq m} U^i$ denote the set of supply nodes for the nodes in path $P$. Let $u_j^i \in U^P$ denote one supply node. Let $n_d(u_j^i)$ denote the total number of demand nodes that $u_j^i$ supports. Let $P_j^i$ denote the set of nodes in $P$ that are supported by $u_j^i$ and $|P_j^i| = n_d(u_j^i, P) \leq n_d(u_j^i)$. We follow a similar method in \cite{gatterbauer2014oblivious} to prove the claim.

  Given the realizations of $U^P \setminus u_j^i = U_S \cup U_F$, where $U_S$ denote the survived nodes and $U_F$ denote the failed nodes, there are three possibilities. First, each node in $P$ has at least one survived supply node in $U_S$, and $P$ survives regardless of the state of $u_j^i$. Second, there exists at least one node in $P$ whose supply nodes are all in $U_F$. Thus, $P$ fails regardless of the state of $u_j^i$. Third, $P$ survives if and only if $u_j^i$ survives.

  The last case occurs if for some nodes in $P_j^i$, all the other supply nodes have failed except $u_j^i$. The probability that $P$ survives is given by $1 - p(u_j^i)$.
  By replacing $u_j^i$ with $n_d(u_j^i, P)$ \emph{distinct} nodes, each of which supports a node in $P_j^i$ and fails independently with probability $1 - (1 - p(u_j^i))^{1/n_d(u_j^i)}$, the probability that all the $n_d(u_j^i, P)$ nodes survive is $(1 - p(u_j^i))^{n_d(u_j^i, P)/n_d(u_j^i)} \geq 1 - p(u_j^i)$. In this case, each node in $P_j^i$ has at least one survived supply node, and the path $P$ survives.

  Thus, by the law of total probability, the probability that $P$ survives never decreases after the above replacement. 

  After repeatedly replacing each supply node that supports multiple demand nodes by distinct nodes, each of which supports a single demand node and fails independently with the specified probability, the demand node failures become independent. Let $\tilde p(v_i)$ denote the failure probability of a demand node $v_i$ after the replacement of supply nodes. The failure probability of a path can be computed efficiently as $1-\prod_{v_i \in P}(1-\tilde p(v_i))$, and is a lower bound on the failure probability of the same path in the original problem.
\end{proof}

Let $n_d$ denote the maximum number of demand nodes that a supply node supports, and let $n_s$ denote the maximum number of supply nodes for a demand node. The following lemma bounds the ratio between the upper and lower bounds. Its proof can be found in the Appendix.
\begin{lem} \label{th:bound}
  For any path, the ratio of the upper bound on its failure probability obtained in Lemma \ref{th:upper} to the lower bound obtained in Lemma \ref{th:lower} is at most $(n_d)^{n_s}$.
\end{lem}
\begin{proof}
  We aim to prove
  $$ \frac{1 - \prod_{v_i \in P}(1 - p(v_i))}{1 - \prod_{v_i \in P}(1 - \tilde p(v_i))} \leq n_d^{n_s},$$
  given
  \begin{equation}\label{eq:failure}
    1 - p(v_i) \geq (1 - \tilde p(v_i))^{n_d^{n_s}},
  \end{equation}
  which will be proved in Lemma \ref{th:inequality0}.

  Given $p(v_i),\tilde p(v_i) \in (0,1)$, with Eq. (\ref{eq:failure}),
  \begin{equation}\label{eq:intermed}
    1 - \prod_{v_i \in P}(1 - p(v_i)) \leq 1 - \prod_{v_i \in P} (1 - \tilde p(v_i))^{n_d^{n_s}}.
  \end{equation}
  and
  $$ 0 < \prod_{v_i \in P} (1 - \tilde p(v_i)) < 1.$$

  Moreover, let
  $$ f(x) = 1 - x^{n_d^{n_s}} - n_d^{n_s}(1 - x).$$
  Since $f(1) = 0$ and
  $$ f'(x) = n_d^{n_s}(1 - x^{n_d^{n_s} - 1}) \geq 0,$$
  for $0 < x \leq 1$ and $n_d^{n_s} \geq 1$, $f(x)$ is an increasing function, and
  $$ f(x) \leq 0,$$
  for $0 < x \leq 1$ and $n_d^{n_s} \geq 1$.

  Let $x =  \prod_{v_i \in P} (1 - \tilde p(v_i)) \in (0,1)$, since $f(x) \leq 0$,
  we obtain that
  $$ 1 - \prod_{v_i \in P} (1 - \tilde p(v_i))^{n_d^{n_s}} \leq {n_d^{n_s}}(1 - \prod_{v_i \in P} (1 - \tilde p(v_i))).$$
  With Eq. (\ref{eq:intermed}),
  $$ 1 - \prod_{v_i \in P}(1 - p(v_i)) \leq {n_d^{n_s}}(1 - \prod_{v_i \in P} (1 - \tilde p(v_i))).$$
  The claim is proved.
\end{proof}

\section{Finding the most reliable path}
\label{sc:search}
In this section, we aim to compute the most reliable path between a source-destination pair $s,t \in V$ in $G(V,E,\ssl_V)$. We first prove that it is NP-hard to approximately compute the most reliable path. We then develop an algorithm to compute the most reliable path when the supply nodes fail independently with an identically small probability, and finally develop an approximation algorithm under arbitrary failure probabilities.

\emph{Hardness of approximation:} Although the failure probability of any given path can be approximated to within factor $1 \pm \epsilon$ for any $\epsilon > 0$, it is NP-hard to obtain an \st path whose failure probability is less than $1 + \epsilon$ times the optimal for a small $\epsilon$. The proof can be found in the Appendix.
\begin{restatable}{thm}{thmNPPath}
  Computing an \st path whose failure probability is less than $1 + \epsilon$ times the failure probability of the most reliable \st path is NP-hard for $\epsilon < 1/m$, where $m$ is the maximum path length.
\end{restatable}

\subsection{Small and identical failure probability}
If every supply node fails independently with an identically small probability, there are two reliability indicators: $n_s^{\min}$ and $\bar m$. Recall that $n_s^{\min}$ is the minimum number of supply nodes for a node in the path, and that $\bar m$ is the number of combinations of $n_s^{\min}$ supply node failures that disconnect the path.
With the two indicators, the path failure probability can be approximated to within a multiplicative factor $1 \pm \epsilon$ by $\bar m p^{n_s^{\min}}$, under the condition that $p \leq \epsilon/m$. Moreover, the indicator $n_s^{\min}$ is more important (and has a higher priority to be optimized) than $\bar m$. We next develop algorithms to optimize the two indicators.

Given a graph $G(V,E,\ssl_V)$ and a pair of nodes $(s,t)$, the problem of computing an \st path with the maximum $n_s^{\min}$ can be formulated as the \emph{maximum capacity path} problem, where the capacity of a node equals the number of its distinct supply nodes and the capacity of a path is the minimum node capacity along the path. The maximum capacity path can be obtained by a modified Dijkstra's algorithm, and can be obtained in linear time \cite{punnen1991linear}.

However, it is NP-hard to minimize $\bar m$, 
even in the special case where every demand node has a single supply node. The result follows from the NP-hardness of computing a path with the minimum colors in a colored graph \cite{yuan2005minimum}.

We develop an integer program to compute the path $P$ with the minimum $\bar m$, under the condition that $n_s^{\min}(P) = \min_{v_i \in P} n_s(v_i)$ is maximized. The following pre-processing reduces the size of the integer program. First, compute $k = \max_{P \in \mathcal{P}} n_s^{\min}(P)$, where $\mathcal{P}$ is the set of all the \st paths, using the linear-time maximum capacity path algorithm. Then, remove all the nodes that have fewer than $k$ distinct supply nodes and their attached edges, and denote the remaining graph by $G'(V',E',\ssl_{V'})$. The removed nodes and edges will not be used by the optimal path. Let $V'' \subseteq V'$ denote the nodes among which each has exactly $k$ distinct supply nodes. We aim to find a path $V_P$ where the number of distinct supply node sets for $V_P \cap V''$ is minimized.

Let $S_i$ denote the set of supply nodes of $i \in V''$. Let $\ssl_{V''}$ denote the union of these sets. Let $x_{ij}$ denote the flow variable which takes a positive value if and only if edge $(i,j)$ belongs to the selected path. An \st path is identified by constraint (\ref{prog:flow}). A node $i$ is on the selected path if at least one of $x_{ij}$ and $x_{ji}$ is positive. Let $h(S_i)$ denote whether removing supply nodes $S_i$ disconnects the selected path. If a node $i$ is on the selected path and has $k$ supply nodes, then $h(S_i)$ must be one, guaranteed by constraint (\ref{prog:supply}). All the other nodes either do not belong to the selected path or have more than $k$ supply nodes, and their supply node failures are not considered. The objective minimizes $\bar m$, which is the number of combinations of $k$ supply node failures that disconnect the path.

\begin{eqnarray}
\text{min} \hspace{-6mm} && \sum_{S_i \in \ssl_{V''}}  h(S_i) \label{prog:comb}\\
\text{s.t.} \hspace{-6mm} &&  \sum_{\{j|(i,j) \in E'\}} \hspace{-3mm} x_{ij} - \hspace{-4mm} \sum_{\{j|(j,i) \in E'\}} \hspace{-3mm} x_{ji} = \left\{
\begin{array}{l l l}
1, ~\text{if} ~i=s,\\
-1, ~\text{if} ~i=t,\\
0, ~\text{otherwise.}
\end{array}
\right.  \label{prog:flow} \\
&& \sum_{\{j|(i,j) \in E'\}} \hspace{-3mm} x_{ij} + \hspace{-4mm} \sum_{\{j|(j,i) \in E'\}} \hspace{-3mm} x_{ji} \leq 2 h(S_i), ~ \forall i \in V'' \setminus s,t, \label{prog:supply} \\
&& x_{ij} \geq 0, ~~~  \forall (i,j)\in E', \nonumber \\
&& h(S_i) = \{0,1\}, ~~~  \forall S_i \in \ssl_{V''}. \nonumber
\end{eqnarray}

\subsection{Arbitrary failure probability}
If nodes $\tilde V$ in a graph $\tilde G( \tilde V, \tilde E)$ fail independently, the probability that a path survives is the product of the survival probabilities of nodes along the path. The most reliable path can be obtained by the classical shortest path algorithm, by replacing the length of traversing a node $\tilde v_i$ by $- \ln(1-p(\tilde v_i))$, where $p(\tilde v_i)$ is the failure probability of $\tilde v_i$. It is easy to see that the length of a path $P$ is $\sum_{\tilde v_i \in P} - \ln(1 - p(\tilde v_i)) =  - \ln \prod_{\tilde v_i \in P}(1 - p(\tilde v_i))$. The shortest path has the smallest failure probability $1 - \prod_{\tilde v_i \in P}(1 - p(\tilde v_i))$.

Compared with the above simple model, the difficulty in obtaining the most reliable \st path in interdependent networks is the failure correlations of nodes $V \subseteq G(V,E,\ssl_V)$. The failure probability of a path can no longer be characterized by $1 - \prod_{v_i \in P}(1 - p(v_i))$. Moreover, let $s - \dots - v_i - \dots - t$ be the most reliable \st path. The sub-path $s - \dots - v_i$ may not be the most reliable path between $s$ and $v_i$. Thus, the label-correction approach in dynamic programming (\eg, Dijkstra's algorithm) cannot be used, even though the failure probability of a given path can be approximated.

Given the bounds obtained in the previous section, we propose Algorithm \ref{alg:minfailure} to compute a path whose failure probability is within $(n_d)^{n_s}$ times the optimal failure probability. Recall that the bounds on path survival probability are the product of (original or new) node survival probabilities, which exactly match the path survival probability in the case of independent node failures.
\begin{algorithm}[h]
\caption{An approximation algorithm to compute a reliable \st path in $G(V,E,\ssl_V)$.}
\begin{enumerate}
\item For each $v_i \in V$, compute $\tilde p(v_i)$ as follows. Let $u_j^i$ be a supply node of $v_i$ with failure probability $p(u_j^i)$. If $u_j^i$ supports $n_d(u_j^i)$ nodes, let $\tilde p(u_j^i) = 1 - (1 - p(u_j^i))^{1/n_d(u_j^i)}$. Let $\tilde p(v_i)$ be the failure probability of $v_i$ if $u_j^i$ fails independently with probability $\tilde p(u_j^i)$.
\item Compute the most reliable \st path assuming that $v_i$ fails independently with probability $\tilde p(v_i)$. The most reliable path can be obtained by a standard shortest path algorithm (\eg, Dijkstra's algorithm), by letting $-\ln(1 - \tilde p(v_i))$ be the length of traversing node $v_i$.
\end{enumerate}
\label{alg:minfailure}
\end{algorithm}
\begin{thm}
  The failure probability of the path obtained by Algorithm \ref{alg:minfailure} is at most $(n_d)^{n_s}$ times the failure probability of the most reliable \st path under arbitrary supply node failure probabilities.
\end{thm}
\begin{proof}
  Let the path obtained by Algorithm \ref{alg:minfailure} be $P'$ and let the path with the minimum failure probability be $P^*$. Let $p(P')$ and $p(P^*)$ denote their failure probabilities. Moreover, let $\tilde p(P')$ and $\tilde p(P^*)$ denote their failure probabilities by assuming that each node $v_i$ fails independently with probability $\tilde p(v_i)$. We have $p(P') \leq n_d^{n_s} \tilde p(P') \leq n_d^{n_s} \tilde p(P^*) \leq n_d^{n_s} p(P^*)$, where the first inequality follows from Lemma \ref{th:bound} and the last inequality follows from Lemma \ref{th:lower}.
\end{proof}

\begin{rmk}
  If $n_s = 1$ and every supply node fails independently with an identically small probability, our result reduces to the following result in the classical shared risk group model: The number of risks associated with the shortest path is at most $n_d$ times the number of risks associated with the minimum-risk path \cite{lee2010diverse}.
\end{rmk}

\section{Reliability of a pair of paths}
\label{sc:pair}
To study diverse routing in interdependent networks, we consider the simplest case of two \st paths in this section. Given that computing the failure probability of a single path is $\#P$ hard if every node has more than one supply node, it is also $\#P$ hard to compute the failure probability of two paths\footnote{Meanwhile, it is still simple to compute the failure probability of two paths if every node has a single supply node, by first computing the probability that the first path fail, and then computing the probability that the second path fail while the first path does not fail (\ie, none of the supply nodes of the first path fail), both in polynomial time, and summing the two probabilities.}. To see this, note that if two paths have the same number of nodes and each node in the first path has identical supply nodes as its corresponding node in the second path, then the probability that both paths fail equals the probability that a single path fails. Fortunately, we are still able to obtain $1 \pm \epsilon$-approximation of the failure probability in polynomial time.

\subsection{Small and identical failure probability}
A central concept in diverse routing is the \emph{disjoint paths} or \emph{risk disjoint paths} \cite{hu2003diverse, yuan2005minimum, coudert2007shared}. In the classical shared risk group model, if every risk occurs independently with an identically small probability $p = o(1/m^2)$, the probability that two paths fail is $\Theta(f(m) p^2)$ if they are risk disjoint and $\Theta(f(m) p)$ if they share one or more risks, where $m$ is the maximum path length and $f(m)$ is a function of $m$. Thus, risk-disjointness characterizes the order of the reliability of two paths. In interdependent networks where every demand node has multiple supply nodes, if nodes in $P^1$ do not share any supply nodes with nodes in $P^2$, then $P^1$ and $P^2$ are risk disjoint. However, risk-disjointness does not suffice to characterize the reliability of two paths, for the following two reasons. First, the failure probability of a demand node depends on the number of supply nodes for it, which is not related to risk-disjointness. Second, if $P^1$ and $P^2$ share some supply nodes, the failure probability depends further on the maximum number of supply node failures that the two paths can withstand. To study the reliability of two paths in interdependent networks, we define \emph{$d$-failure resilient paths} as follows.

\begin{defn}
Two paths are $d$-failure resilient if removing any $d$ supply nodes would not disconnect both paths.
\end{defn}

\begin{rmk}
In the classical graph model $\tilde G(\tilde V,\tilde E)$, two disjoint paths are 1-failure resilient while two overlapping paths are 0-failure resilient. In the classical shared risk group model, two risk disjoint paths are 1-failure resilient while two paths that share risks are 0-failure resilient. Two paths can never be more than one failure resilient. Thus, the disjointness or risk-disjointness suffices to characterize (the order of) the reliability of two paths in these models.
\end{rmk}

\subsubsection{Evaluation of failure probability}
Consider two paths $P^1 = s - v_1^1 - v_2^1 - \dots - v_{m_1}^1 - t$, $P^2 = s - v_1^2 - v_2^2 - \dots - v_{m_2}^2 - t$ between a pair of nodes $(s,t)$. We study the event that at least one node in $P^1$ and at least one node in $P^2$ both fail. Let $F_i^k$ denote the event that all the supply nodes of $v_i^k$ fail, and let $F^k$ denote the event that the $k$-th path fails, $k \in \{1,2\}$. Then $F^1 \cap F^2 = \cup_{1 \leq i \leq m_1, 1 \leq j \leq m_2} (F_i^1 \cap F_j^2)$. For simplicity of presentation, let $F^{\text{both}}= F^1 \cap F^2$ and $ F_{ij} = F_i^1 \cap F_j^2$. Let $S_{ij}$ denote the union of supply nodes of $v_i^1$ and $v_j^2$.

To decide whether two paths are $d$-failure resilient, we consider the number of supply node failures that lead to the event $F_{ij}$, and denote the number by $d_{ij}$. Then $d = \min_{1 \leq i \leq m_1, 1 \leq j \leq m_2} d_{ij} - 1$. Moreover, let $\bar m$ be the number of pairs of nodes, one from each path, such that each pair of nodes in total have $d + 1$ distinct supply nodes and any two pairs do not have the same set of $d + 1$ supply nodes. (\Ie, $\bar m$ combinations of $d + 1$ supply node failures each disconnect both paths.) The next theorem formalizes the connection between the reliability of two paths and $d$.

\begin{thm} \label{th:two}
If every supply node fails independently with probability $p \leq \epsilon /(m_1 m_2)$, then the probability that two $d$-failure resilient paths with lengths $m_1, m_2$ both fail satisfies $(1 - \epsilon) \bar m p^{d+1} \leq \Pr(F^{\text{both}}) \leq (1 + \epsilon) \bar m p^{d+1}$.
\end{thm}
\begin{proof}
First consider the events $F_{ij} = F_i^1 \cap F_j^2$, $1\leq i \leq m_1, 1 \leq j \leq m_2$. Let $S_{ij}$ denote the union of supply nodes of $v_i^1$ and $v_j^2$. Then the event $F_{ij}$ occurs if and only if all the nodes $S_{ij}$ fail. By a similar argument as the proof of Theorem \ref{th:singleGeneral}, if $S_{i_1 j_1}$ is a subset of $S_{i_2 j_2}$, then $F_{i_2 j_2}$ occurs only if $F_{i_1 j_1}$ occurs, and $F_{i_2 j_2}$ is redundant. In the following we only consider $\mathcal S = \{S_{ij}|1 \leq i \leq m_1, 1 \leq j \leq m_2,  \text{ none of $S_{ij}$ is a subset of another.}\}$. The cardinality of $\mathcal S$ at most $m_1 m_2$.

By the inclusion-exclusion principle, $\Pr(F^{\text{both}})$ can be computed as follows.
\begin{align}
\Pr(F^{\text{both}}) &= \sum_{S_{ij} \in \mathcal S} \Pr(S_{ij} \text{ fail}) \nonumber \\
& - \sum_{S_{i_1 j_1}, S_{i_2 j_2} \in \mathcal S} \Pr(S_{i_1 j_1} \cup S_{i_2 j_2} \text{ fail}) \nonumber \\
& + \dots + (-1)^{|\mathcal{S}| - 1} \Pr(\cup_{S_{ij} \in \mathcal S} S_{ij} \text{ fail}).
\label{eq:inc2}
\end{align}

Since two paths are $d$-failure disjoint, the number of nodes in $S_{i j}$ is $d + 1$ for some $i\in\{1,\dots,m_1\}, j \in \{1,\dots, m_2\}$ while the number of nodes in all the other $S_{ij}$ is larger than $d + 1$. Let $\mathcal{S}_1 \subseteq \mathcal S$ be the union of the supply node sets, each of which contains $d+1$ nodes, and let $\bar m = |\mathcal{S}_1|$. The first term in Eq. (\ref{eq:inc2}) is at most $\bar m q^{d + 1} + (|\mathcal S| - \bar m) q^{d + 2}$. Therefore,
\begin{align*}
  \Pr(F^{\text{both}}) \leq & \bar m q^{d + 1} + m_1 m_2 q q^{d + 1} \\
  \leq & \bar m q^{d + 1} + \epsilon q^{d + 1},
\end{align*}
if $q \leq \epsilon/(m_1 m_2)$.

We next consider the supply node failures of two pairs of nodes. Recall that $\mathcal{S}_1$ consists of supply node sets that each contain $d+1$ nodes. Let $\mathcal{S}_2 = \mathcal{S} \setminus \mathcal{S}_1$ be the remaining supply node sets that each contain $d+2$ or more nodes. The union of two sets $S_{i_1 j_1} \cup S_{i_2 j_2}$ ($S_{i_1 j_1}, S_{i_2 j_2} \in \mathcal S_1$) contains at least $d+2$ nodes. The union $S_{i_1 j_1} \cup S_{i_2 j_2}$ ($S_{i_1 j_1}, S_{i_2 j_2} \in \mathcal S_2$, or $S_{i_1 j_1} \in \mathcal{S}_1, S_{i_2 j_2} \in \mathcal S_2$) contains at least $d+3$ nodes. The absolute value of the second term is at most ${\bar m \choose 2} q^{d+2} + [{m_1m_2  \choose 2} - {\bar m \choose 2}] q^{d+3}$. To conclude,
\begin{align*}
  \Pr(F^{\text{both}}) \geq & \bar m q^{d+1} - \big(\frac{\bar m^2}{2} q q^{d+1} + \frac{(m_1 m_2)^2}{2} q^2 q^{d+1}\big)\\
  \geq & \bar m q^{d+1} - \epsilon \bar m q^{d+1},
\end{align*}
if $q \leq \epsilon/(m_1 m_2)$.

\end{proof}

\subsubsection{Finding the most reliable pair of paths}
From Theorem \ref{th:two}, we know that the probability that two $d$-failure resilient paths both fail is smaller for larger values of $d$. Moreover, for a fixed $d$, the failure probability is proportional to $\bar m$, the number of combinations of $d+1$ supply node failures that disconnect both paths. We have obtained two reliability indicators for two paths: $d$ and $\bar m$.

Unfortunately, computing the pair of \st paths that have the maximum $d$ and the minimum $\bar m$ are both NP-hard, even in the special case where every demand node has a single supply node. This special case reduces to the classical shared risk group model. In this special case, $d = 1$ if there exist two risk-disjoint paths, and $d = 0$ otherwise. The NP-hardness of determining the existence of two risk-disjoint paths between an \st pair has been proved in \cite{hu2003diverse}. 
Moreover, in this special case, for two paths that share common supply nodes, $\bar m$ is the number of overlapping risks between the two paths (\ie, removing any of the $\bar m$ supply nodes disconnects both paths). The NP-hardness of the least coupled paths problem, which computes a pair of paths that share the minimum number of risks in the classical shared risk group model, has also been proved in \cite{hu2003diverse}.

We develop an integer program to compute a pair of \st paths with the maximum $d$ in $G(V,E,\ssl_V)$. Let variable $x_{ij}^k$ denote whether edge $(i,j)$ is part of the $k$-th path, and let variable $b_i^k$ denote whether node $i$ is part of the $k$-th path, $k \in \{1,2\}$. Same as before, let $S_i$ denote the supply nodes of node $i$. Constraints (\ref{prog:disjoint}) guarantee that two paths are node-disjoint. Notice that these constraints can be dropped if there is no restriction on the physical disjointness of two paths. Constraints (\ref{prog:twofailures}) guarantee that at least $d+1$ supply nodes need to be removed in order for one node in each path to fail (\ie, $b_i^1 = b_j^2 = 1$, $i,j \in V$), where $M$ is a sufficiently large number, \eg, twice the maximum number of supply nodes for a demand node.
\begin{eqnarray}
\text{max} && \hspace{-3mm} d \label{prog:level}\\
\text{s.t.}&& \hspace{-5mm} \sum_{\{j|(i,j) \in E\}} \hspace{-3mm} x_{ij}^k - \hspace{-4mm} \sum_{\{j|(j,i) \in E\}} \hspace{-3mm} x_{ji}^k = \left\{
\begin{array}{l l l}
1, ~\text{if} ~i=s,\\
-1, ~\text{if} ~i=t,\\
0, ~\text{otherwise.}
\end{array}
\right. k \in \{1,2\}, \nonumber \\
&& \hspace{-5mm}  \sum_{\{j|(i,j) \in E\}} \hspace{-3mm} x_{ij}^k + \hspace{-4mm} \sum_{\{j|(j,i) \in E\}} \hspace{-3mm} x_{ji}^k \leq 2 b_i^k, ~~ \forall i \in V, k \in \{1,2\} \nonumber \\
&& \hspace{-3mm} b_{i}^1 + b_{i}^2 \leq 1, ~\forall i \in V \setminus s,t, \label{prog:disjoint} \\
&& \hspace{-5mm} d + 1 \leq |S_i \cup S_j| + M(2 - b_i^1 - b_j^2), \forall i,j \in V \setminus s,t, \label{prog:twofailures} \\
&& \hspace{-5mm} x_{ij}^k \in \{0, 1\},~~\hfill \forall (i,j)\in E, k \in \{1,2\}, \nonumber \\
&& \hspace{-5mm} b_i^k \in \{0,1\}, ~~ \hfill \forall i \in V, k \in \{1,2\}. \nonumber
\end{eqnarray}

A slightly modified integer program suffices to minimize $\bar m$ under the condition that $d$ is maximized. Let $h(S_i \cup S_j)$ denote whether removing the union of supply nodes for $i$ and $j$ disconnects both paths. Constraints (\ref{prog:twofailures2}) guarantee that if $i$ and $j$ belong to two different paths, \ie, $b_i^1 = b_j^2 = 1$, then $h(S_i \cup S_j) = 1$. Otherwise, $h(S_i \cup S_j) = 0$ in the optimal solution. Let a positive value $w(|S_i \cup S_j|)$ denote its \emph{weight}, which is a decreasing function of the cardinality $|S_i \cup S_j|$. We aim to minimize the total weights of supply node failures that disconnect two paths. In order to guarantee that $d$ is maximized, $w(l) / w(l+1)$ should be sufficiently large for any integer $l$, \eg, $|V|^2 / 2$. Since there are at most $|V|(|V|-1)/2$ pairs of nodes, larger $d$ is always preferable and has a higher priority to be optimized over $\bar m$.

\begin{eqnarray}
\text{min} && \hspace{-3mm} \sum_{S_i, S_j \in \ssl_V} w(|S_i \cup S_j|) h(S_i \cup S_j) \label{prog:comb2}\\
\text{s.t.}&& \hspace{-5mm} \sum_{\{j|(i,j) \in E\}} \hspace{-3mm} x_{ij}^k - \hspace{-4mm} \sum_{\{j|(j,i) \in E\}} \hspace{-3mm} x_{ji}^k = \left\{
\begin{array}{l l l}
1, ~\text{if} ~i=s,\\
-1, ~\text{if} ~i=t,\\
0, ~\text{otherwise.}
\end{array}
\right. k \in \{1,2\}, \nonumber \\
&& \hspace{-5mm}  \sum_{\{j|(i,j) \in E\}} \hspace{-3mm} x_{ij}^k + \hspace{-4mm} \sum_{\{j|(j,i) \in E\}} \hspace{-3mm} x_{ji}^k \leq 2 b_i^k, ~~ \forall i \in V, k \in \{1,2\} \nonumber \\
&& \hspace{-3mm} b_{i}^1 + b_{i}^2 \leq 1, ~\forall i \in V \setminus s,t, \nonumber \\
&& \hspace{-3mm} h(S_i \cup S_j) \geq b_i^1 + b_j^2 - 1, ~~ \forall i,j \in V \setminus s,t, \label{prog:twofailures2} \\
&& \hspace{-3mm} x_{ij}^k \in \{0, 1\}, ~~\hfill \forall (i,j)\in E, k \in \{1,2\}, \nonumber \\
&& \hspace{-3mm} h(S_i \cup S_j) \in \{0,1\}, ~~ \hfill \forall i,j \in V. \nonumber
\end{eqnarray}

\subsection{Arbitrary failure probability}
\label{sc:twoDNF}
\subsubsection{Evaluation of failure probability}
We use a similar importance sampling approach to Algorithm \ref{al:sampling} and formulate the problem of computing the failure probability of two paths as a DNF probability problem. A clause $C_{ij}$ represents a pair of nodes $v_i^1$ and $v_j^2$. Literals in $C_{ij}$ represent the union of supply nodes of $v_i^1$ and $v_j^2$. A literal is true if and only if the supply node that it represents fails, and the probability that the literal is true is the same as the supply node failure probability. The disjunction of clauses is true if and only if at least one clause is true, in which case both paths fail because at least one node from each path fails.
The rest of the computation follows the same manner as Algorithm \ref{al:sampling}, by replacing a node in Algorithm \ref{al:sampling} by a pair of nodes. An $(\epsilon, \delta)$-approximation of the failure probability $\Pr(F^{\text{both}})$ can be obtained in $O(m_1^2 m_2^2 n_s \ln(1/\delta) / \epsilon^2)$ time. 

\subsubsection{Finding the most reliable pair of paths}
It is more difficult to find two paths that have the smallest failure probability. Recall Theorem \ref{th:two}. The failure probability of two paths is $\Theta(f(m_1, m_2)p^{d+1})$ if they are $d$-failure resilient when the supply node failure probability $p$ is small, where $f(m_1, m_2)$ is a function of two path lengths. As a corollary of the fact that it is NP-hard to compute two paths that have the maximum level of resilience $d$, it is also NP-hard to compute two paths whose failure probability is within a factor $\alpha$ from the optimal, where $\alpha$ is any function of the network size. Thus, we develop the following heuristic. After computing the failure probability $\tilde p(v_i)$ of a node $v_i$ in Step 1 of Algorithm \ref{alg:minfailure}, let $-\ln(1- \tilde p(v_i))$ be the length of traversing node $v_i$, and compute two node disjoint paths with the minimum total lengths. The two paths can be efficiently obtained using a slightly modified shortest augmenting path algorithm \cite{edmonds1972theoretical}. The computation is outlined in Algorithm \ref{alg:pair}. The reason for the graph transformation in Step 1 is to simplify the computation of a residual graph, to which the shortest augmenting path algorithm can be applied.

\begin{algorithm}[h]
\caption{A heuristic to compute a pair of reliable \st path in $G(V,E,\ssl_V)$.}
\begin{enumerate}
\item Transform $G(V,E,\ssl_V)$ with node failure probabilities to a directed graph $G'$ with edge failure probabilities using the standard approach. (Split every node $v$ into $v_\text{in}$ and $v_\text{out}$. Add a directed edge from $v_\text{in}$ to $v_\text{out}$, which has length $-\ln(1- \tilde p(v_i))$. Add a directed edge from $v_{1\text{out}}$ to $v_{2\text{in}}$ and a directed edge from $v_{2\text{out}}$ to $v_{1\text{in}}$, both with zero length, if an edge exists between $v_1$ and $v_2$ in $G(V,E,\ssl_V)$.)
\item Compute the shortest path $P'_1$ from $s_\text{out}$ to $t_\text{in}$ in $G'$.
\item Compute the residual graph. Remove all the edges in $P'_1$. Add a backward edge from $v'_2$ to $v'_1$ with a negated length if an edge from $v'_1$ to $v'_2$ is part of $P'_1$.
\item Compute the shortest path $P'_2$ from $s_\text{out}$ to $t_\text{in}$ in the residual graph.
\item Combine $P'_1$ and $P'_2$ by cycle cancellation. The two paths become node-disjoint and can be mapped to two paths in $G(V,E,\ssl_V)$.
\end{enumerate}
\label{alg:pair}
\end{algorithm}
\color{black}
\section{Numerical results}
\label{sc:numerical}
We study the robust routing problems in the XO backbone communication network with 60 nodes and 75 edges \cite{xo}, by assuming that the XO nodes are supported by 36 randomly generated supply nodes within the continental US. The XO network topology is depicted in Fig. \ref{fig:xo}, and the supply nodes are marked as triangles. The $x$-axis represents the longitude and the $y$-axis represents the latitude. We do not claim that the XO network needs supply from these randomly generated points, and we use this example only to provide a visualization of the robust routing problems using available data. 
\begin{figure}[h]
\begin{centering}
\leavevmode\includegraphics[width=\linewidth]{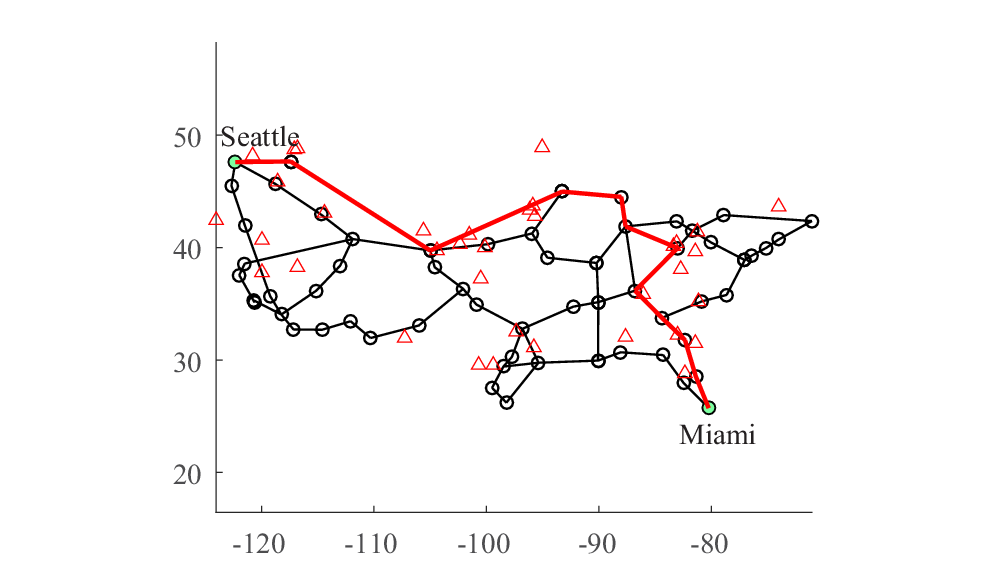}
\caption{Topology of the XO network and randomly generated triangle supply nodes. The most reliable Seattle-Miami path is colored red under the condition that supply node failure probability is small and identical and every XO node depends on two nearest supply nodes.}
\label{fig:xo}
\end{centering}
\end{figure}

First, we assume that every XO node depends on two nearest supply nodes and every supply node fails independently with probability $10^{-2}$. Since the supply node failure probability is small and identical, we are able to obtain the most reliable path and pair of paths by optimizing the reliability indicators using integer programs.

To identify the most reliable path, since $n_s^{\min}(P) = 2$ for any path $P$, we only need to compute a path with the minimum $\bar m$ using the integer program in Section \ref{sc:search}. The most reliable path is colored red in Fig. \ref{fig:xo}, for which $\bar m = 8$. To evaluate the path failure probability, by Corollary \ref{th:single}, setting $\epsilon = 4 \times 10^{-2}$, $\Pr(F) \in [7.68 \times 10^{-4}, 8 \times 10^{-4}]$. To compare, using Algorithm \ref{al:sampling}, we obtain $7.9686 \times 10^{-4}$ as a ($1 \pm 0.01$)-approximation of the path failure probability with probability 0.99. These results suggest that the two reliability indicators ($n_s^{\min}, \bar m$) well characterize the path failure probability when the supply node failure probability is small and identical.

We compute the most reliable pair of paths connecting Seattle-Miami using the integer programs in Section \ref{sc:pair}. The two paths are plotted in Fig. \ref{fig:xotwo}, and they are 1-failure resilient ($d = 1$). The failure probability of both paths is approximately $1.0388 \times 10^{-4}$. In contrast, the most reliable pair of paths connecting Seattle-Denver are 3-failure resilient and their failure probability is approximately $2.9800 \times 10^{-8}$. Thus, the level of resilience well indicates the reliability of two paths.
\begin{figure}[h]
\begin{centering}
\leavevmode\includegraphics[width=\linewidth]{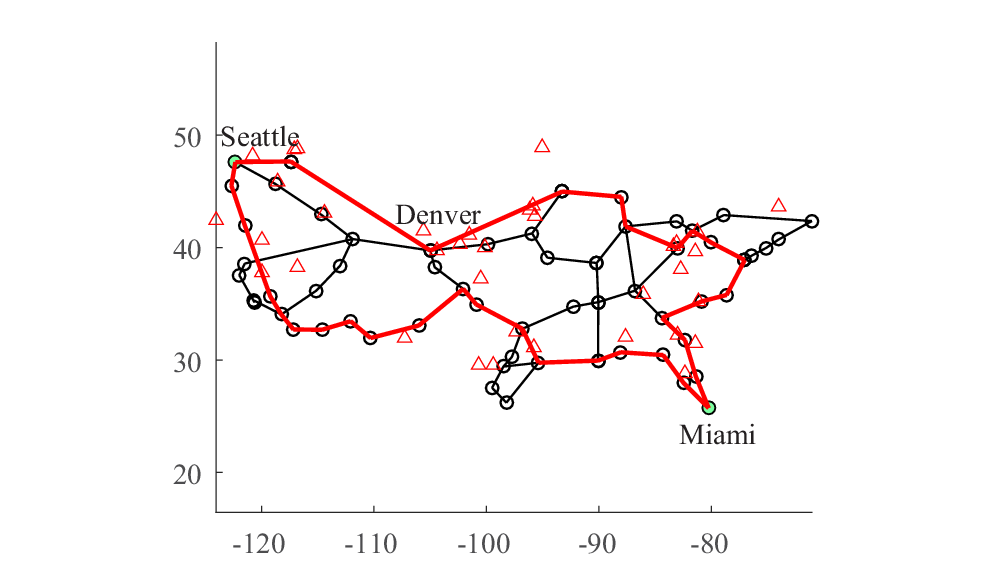}
\caption{The most reliable pair of paths between Seattle-Miami are colored red, under the condition that supply node failure probability is small and identical and every XO node depends on two nearest supply nodes.}
\label{fig:xotwo}
\end{centering}
\end{figure}

Next, we assume that an XO node depends on $N_s$ randomly chosen supply nodes, where $N_s$ is uniformly chosen among 1, 2, and 3. Let the failure probability of each supply node be uniformly and independently chosen from $[0.005, 0.015]$. We use Algorithm \ref{alg:minfailure} to obtain a reliable path connecting Seattle-Miami. Averaged over 10 trials, the path failure probability is approximately $2.1032 \times 10^{-2}$, while the lower bound on the failure probability of the most reliable path is $5.2365 \times 10^{-3}$. The obtained path has failure probability around four times the lower bound. Moreover, by using the heuristic to find a pair of paths, the paths have average failure probability $3.9732 \times 10^{-3}$, which improves the reliability of a single path.

We compare the performance of the heuristic (Algorithm \ref{alg:pair}) with the optimal pair of paths. Since it is difficult to obtain the optimal pair of paths under arbitrary failure probabilities, we use the integer program (\ref{prog:comb2}), under the condition that supply nodes fail independently with probability $10^{-2}$. If every XO node depends on two nearest supply nodes, the failure probabilities of two optimal paths and two paths obtained by the heuristic are approximately $1.0388 \times 10^{-4}$ and $1.0773 \times 10^{-4}$, respectively. If every XO node depends on three nearest supply nodes, the failure probability of two optimal paths and two paths obtained by the heuristic are approximately $1.0200\times10^{-6}$ and $1.0508 \times 10^{-6}$, respectively. These experiments validate the performance of our heuristic algorithm. 

Finally, we report the running times of the algorithms, executed in a workstation that has an Intel Xeon Processor (E5-2687W v3) and 64GB RAM. The integer programs that find the most reliable path and pair of paths (under small and identical supply node failure probability) can both be solved within 1 second. The approximation algorithm to find a reliable path and the heuristic to find a pair of paths (under arbitrary failure probabilities) can both be solved within 0.1 second. The evaluation of the failure probability of one path or a pair of paths by Algorithm \ref{al:sampling} takes several minutes, by setting $\epsilon = \delta = 0.01$. Thus, the algorithms (integer programs and Algorithms \ref{alg:minfailure} and \ref{alg:pair}) can be used to find reliable routes in realistic size networks.
\section{Conclusion}
\label{sc:conclusion}
We studied the robust routing problem in interdependent networks. We developed approximation algorithms to compute the path failure probability, and identified reliability indicators for a path, based on which we develop algorithms to find the most reliable route in interdependent networks. We also studied diverse routing in interdependent networks, and developed approximation algorithms to compute the probability that two paths both fail and to find two reliable paths. Our work extends the shared risk group models, and provides a new framework to study robust routing problems in interdependent networks.

\section*{Appendix}
\subsection{Computational complexity}

In this section, we prove the complexity of computing the path failure probability and finding the most reliable path.
\thmSharpP*

\begin{proof}
The problem of computing the path failure probability can be reduced from a \emph{monotone DNF counting} problem. A monotone DNF counting problem aims to compute the number of satisfying assignments of literals, for a DNF formula that has no negated literals. The monotone DNF counting problem is $\# P$-hard, even if every clause contains two literals~\cite{valiant1979complexity}. (The original paper \cite{valiant1979complexity} considers conjunctive normal form counting, monotone 2-CNF(SAT). It is easy to see the equivalence between the monotone 2-DNF and monotone 2-CNF by applying De Morgan's law and negating all the literals.)

Given a monotone DNF counting problem, construct a path as follows. Each node in the path represents a clause, and its supply nodes represent the literals in the clause. (See Fig. \ref{fig:DNF} for an example.)
\begin{figure}[h]
\begin{centering}
\leavevmode\includegraphics[width=0.8\linewidth]{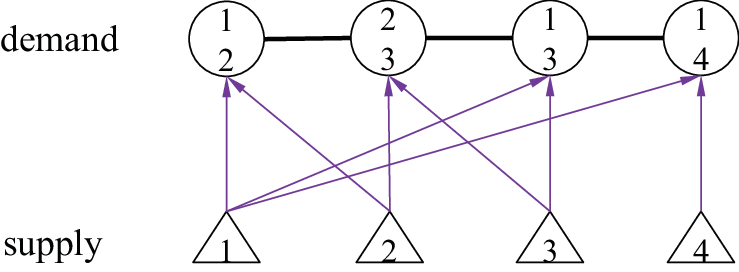}
\caption{A path constructed from a monotone DNF formula $(x_1 \land x_2) \lor (x_2 \land x_3) \lor (x_1 \land x_3) \lor (x_1 \land x_4)$.}
\label{fig:DNF}
\end{centering}
\end{figure}

If every supply node fails independently with probability $1/2$, then the path failure probability is $N / 2^m$, where $m$ is the total number of supply nodes (literals), and $N$ is the number of combinations of supply node failures that lead to the failure of at least one node, which equals the number of satisfying assignments for the DNF formula. Thus, the failure probability of a path under $p = 1/2$ gives an answer to the monotone DNF counting problem. To conclude, computing the path failure probability is $\# P$-hard if every node has two or more supply nodes.
\end{proof}

\begin{rmk}
We further consider the complexity of computing the path reliability, with additional restrictions on the maximum number of demand nodes that any supply node supports. If every supply node supports at most two demand nodes, in addition to the restriction that every demand node has at most two supply nodes, then the failure probability can be computed in polynomial time, when every supply node fails independently with an identical probability. The computation follows from the algorithm in \cite{roth1996hardness}, which relates the number of satisfying assignments of a DNF formula, where each literal appears at most twice and each clause contains two literals, to the number of independent sets in a graph with node degree at most two. Nevertheless, if every supply node supports three or more demand nodes, the computation becomes $\#P$-hard even if every demand node has at most two supply nodes, because counting the number of independent sets in a graph with node degree three is $\#P$-hard \cite{greenhill2000complexity}.
\end{rmk}

\thmNPPath*
\begin{proof}
  We prove that computing such a path is NP-hard even in the following restricted case. Consider a graph where every node has a single supply node, and every supply node fails independently with an identically small probability $p = o(1/m^2)$. The failure probability of a path supported by $\bar m \leq m$ supply nodes is $1 - (1 - p)^{\bar m} = \bar m p + o(p)$.

  Suppose that the most reliable path is supported by $\bar m_{\min}$ supply nodes and has failure probability $p_{\min} = \bar m_{\min} p + o(p)$. For $\epsilon < 1/m \leq 1/\bar m_{\min}$, a path with failure probability strictly smaller than $(1+\epsilon)p_{\min} < (\bar m_{\min} + 1)p + o(p)$ is supported by exactly $\bar m_{\min}$ supply nodes. Computing a path that is supported by the minimum number of supply nodes in this example is NP-hard, which is known as the minimum color path problem in \cite{yuan2005minimum}. Therefore, computing a path that has failure probability within $1 + \epsilon$ times the optimal is NP-hard for $\epsilon < 1/m$.
\end{proof}

\subsection{Approximating the path failure probability by importance sampling}
In this section, we prove the correctness of Algorithm \ref{al:sampling} and Theorem \ref{th:time}.
\begin{lem} \label{th:sample}
The path failure probability is given by $E[I] \sum_{1 \leq k \leq m} \prod_{1 \leq j \leq n_s(v_k)} p(u_j^k)$.
\end{lem}
\begin{proof}
Let $V_s$ denote the set of all supply nodes. Let $U$ denote a set of failed supply nodes that lead to the failure of at least one of $\{ v_1, v_2, \dots, v_{m} \}$ (\ie, the failure of the path). Let $\mathcal{U} = \{U_1, U_2, \dots, U_R \}$ denote all the sets of supply node failures that lead to the failure of the path. Let $p(u)$ denote the failure probability of node $u$. Let $\Pr(\text{exactly}~U_r~\text{fail})$ denote the probability that supply nodes $U_r$ fail and all the other supply nodes $V_s \setminus U_r$ do not fail. Since the events that $U_r$ fail while the others do not fail are mutually exclusive for different $r$, the path failure probability is given by
\begin{align}
\Pr(F) &= \sum_{1 \leq r \leq R} \Pr(\text{exactly}~U_r~\text{fail}) \nonumber \\
&= \sum_{1 \leq r \leq R} \prod_{u \in U_r} p(u) \prod_{u \in V_s \setminus U_r}(1 - p(u)). \label{eq:left1}
\end{align}

Let $\{u_j^k | j = 1,\dots, n_s(v_k)\}$ denote the set of supply nodes of $v_k$. Let $m_r$ denote the number of demand node failures, among $\{v_1, v_2, \dots, v_{m} \}$, if supply nodes $U_r$ fail. Then, by summing the failure probabilities of demand nodes $ p(v_k) = \prod_{1 \leq j \leq n_s(v_k)} p(u_j^k)$, $k = 1,\dots,m$, the probability that supply nodes $U_r$ fail is counted $m_r$ times.
\begin{align}
\sum_{1 \leq k \leq m} & \prod_{1 \leq j \leq n_s(v_k)} p(u_j^k) = \sum_{1 \leq r \leq R} m_r \Pr(\text{exactly}~U_r~\text{fail}) \nonumber \\
& = \sum_{1 \leq r \leq R} m_r \prod_{u \in U_r} p(u) \prod_{u \in V_s \setminus U_r}(1 - p(u)).
\label{eq:left2}
\end{align}

We now construct the relationship between the left hand sides of Eq. (\ref{eq:left1}) and Eq. (\ref{eq:left2}) using $E[I]$. In Algorithm \ref{al:sampling}, the value of $I$ in Step 4 depends on both $v_i$ (obtained in Step 2) and $U$ (obtained in Step 3). In the remainder of the proof, we first compute the probability that a specific $U$ is obtained (in an iteration of the main loop), and then compute $\Pr(I = 1 | U)$ (\ie, the probability that $v_i$ is the first failed node given that $U$ fail). As a consequence, $\Pr[I = 1]$ can be determined using the law of total probability.

Consider an iteration of the main loop. Let $\Pr(U_r)$ denote the probability that $U_r$ is obtained in Step 3 of Algorithm \ref{al:sampling}. Let $r(t), t = 1,\dots, m_r$ denote the indices of failed nodes $v_{r(t)}$ among $\{v_1,v_2,\dots,v_m\}$ if $U_r$ fail ($m_r \leq m$). We have
\begin{align}
\Pr(U_r) &= \sum_{1 \leq t \leq m_r} \bigg( \frac{ \prod_{1 \leq j \leq n_s(v_{r(t)})} p(u_j^{r(t)})}{\sum_{1 \leq k \leq m} \prod_{1 \leq j \leq n_s(v_k)} p(u_j^k)} \nonumber \\
& \times \hspace{-5mm} \prod_{u \in U_r \setminus \{ u_j^{r(t)}, 1 \leq j \leq n_s(v_{r(t)}) \}} \hspace{-7mm} p(u) \prod_{u \in V_s \setminus U_r} (1 - p(u)) \bigg) \label{eq:total} \\
&= \sum_{1 \leq t \leq m_r} \frac{\prod_{u \in U_r } p(u) \prod_{u \in V_s \setminus U_r} (1 - p(u))} {\sum_{1 \leq k \leq m} \prod_{1 \leq j \leq n_s(v_k)} p(u_j^k)} \label{eq:uniform} \\
&= \frac{m_r \prod_{u \in U_r } p(u) \prod_{u \in V_s \setminus U_r} (1 - p(u))}{\sum_{1 \leq k \leq m} \prod_{1 \leq j \leq n_s(v_k)} p(u_j^k)}.
\label{eq:sample}
\end{align}
To see this, note that there are $m_r$ choices of $\{v_{r(t)}| t = 1,\dots, m_r \}$ which may lead to $U_r$. In Eq. (\ref{eq:total}), the first term (in the product) is the probability of choosing $v_{r(t)}$ and setting its supply nodes $U^{r(t)} = \{ u_j^{r(t)} | j = 1, \dots, n_s(v_{r(t)}) \}$ to be failed; the second term is the probability that $U_r \setminus U^{r(t)}$ fail; the last term is the probability that the remaining supply nodes $V_s \setminus U_r$ do not fail.

Eq. (\ref{eq:uniform}) implies that $v_{r(t)}, t \in \{1,\dots, m_r\}$ contribute equally to the occurrence of $U_r$. Namely,
\begin{align*}
  ~~~~& \Pr(v_{r(t)} \text{ has been chosen in Step 2} |U_r) \\
  = & \bigg( \frac{ \prod_{1 \leq j \leq n_s(v_{r(t)})} p(u_j^{r(t)})}{\sum_{1 \leq k \leq m} \prod_{1 \leq j \leq n_s(v_k)} p(u_j^k)} \nonumber \\
  & \times \hspace{-3mm} \prod_{u \in U_r \setminus \{ u_j^{r(t)}, 1 \leq j \leq n_s(v_{r(t)}) \}} \hspace{-7mm} p(u)  \prod_{u \in V_s \setminus U_r} (1 - p(u)) \bigg)
  \bigg/ \Pr(U_r) \nonumber \\
  = & 1/m_r,
\end{align*}
for $t \in \{1,\dots, m_r\}$.


Given $U^r$, the probability that $v_{r(1)}$ has been chosen in Step 2 of Algorithm \ref{al:sampling} is $1/m_r$. Thus, $\Pr[I = 1|U_r] = 1/m_r$.
By the law of total probability,
\begin{align*}
\Pr[I = 1] &= \sum_{1 \leq r \leq R} \Pr[I = 1|U_r] \Pr(U_r) \\
&= \frac{\sum_{1 \leq r \leq R} \prod_{u \in U_r } p(u) \prod_{u \in V_s \setminus U_r} (1 - p(u))}{\sum_{1 \leq k \leq m} \prod_{1 \leq j \leq n_s(v_k)} p(u_j^k)}.
\end{align*}

Since $I$ is an indicator variable, $E[I] = \Pr (I = 1)$.
\begin{align*}
& ~ E[I] \sum_{1 \leq k \leq m} \prod_{1 \leq j \leq n_s(v_k)} p(u_j^k) \\
&= \sum_{1 \leq r \leq R} \prod_{u \in U_r } p(u) \prod_{u \in V_s \setminus U_r} (1 - p(u)) \\
&= \Pr(F).
\end{align*}

\end{proof}

Next, we prove that $E[I]$ can be estimated accurately within $3 m \ln (2/\delta)/\epsilon^2$ iterations.
\begin{lem}
$$\Pr\Big(\Big|\frac{E[I] - b/a}{E[I]}\Big| \geq \epsilon E[I]\Big) \leq \delta,$$
where $a = 3 m \ln (2/\delta)/\epsilon^2$ is the number of iterations of the main loop of Algorithm \ref{al:sampling}, $b$ is the number of observations of $I = 1$, and $0 < \epsilon < 1$.
Namely, by repeating $a = 3 m \ln (2/\delta)/\epsilon^2$ times, one obtains an $(\epsilon, \delta)$-approximation of $E[I]$.
\end{lem}
\begin{proof}
The proof is based on the Chernoff inequality and is a standard result in estimation theory. From the proof of Lemma \ref{th:sample}, we know that $\Pr(I = 1|U^r) = 1/m^r \geq 1/m$ for all $U^r$. Thus, $\Pr(I = 1) \geq 1/m$. To estimate $E[I]$ within $1 \pm \epsilon$ accuracy ($0 < \epsilon < 1$), let the number of trials be $a = 3 m \ln (2/\delta)/\epsilon^2$.
\begin{align*}
& ~ \Pr \bigg( \bigg| \sum_{1 \leq i \leq a} I_i - \sum_{1 \leq i \leq a} E[I_i] \bigg|
 \geq \epsilon \sum_{1 \leq i \leq a} E[I_i] \bigg)\\
&\leq \exp(-\frac{\epsilon^2 \sum_{1 \leq i \leq a} E[I_i]}{2}) + \exp(-\frac{\epsilon^2 \sum_{1 \leq i \leq a} E[I_i]}{3}) \\
& \leq 2 \exp(-\frac{\epsilon^2 a / m}{3})
= 2 \exp(-\frac{3 \ln(2/\delta)}{3}) \\
& \leq \delta.
\end{align*}
\end{proof}

Finally we prove the time complexity of the algorithm.
\thmtime*
\begin{proof}

  Consider an iteration of the main loop of Algorithm \ref{al:sampling}. In Step 2, obtaining $v_i$ takes $O(m n_s)$ time. In Step 3, obtaining $U$ takes $O(m n_s)$ time, because the total number of supply nodes is at most $O(m n_s)$. In Step 4, testing whether $v_i$ is the first failed node under the failure of $U$ takes $O(m n_s)$ time, given that checking whether a node fail takes $O(n_s)$ time and there are at most $m$ nodes in the path.

  Since $3 m \ln (2/\delta)/\epsilon^2$ iterations are sufficient, the total running time of Algorithm \ref{al:sampling} is $O(m^2 n_s \ln(1/\delta) / \epsilon^2)$.
\end{proof}
\subsection{Bounds on path failure probability}

\begin{lem} \label{th:inequality}
  \begin{equation*}
    1 - (1 - p_1 p_2)^{\alpha \beta} \leq [1 - (1 - p_1)^\alpha] [1 - (1 - p_2)^\beta],
  \end{equation*}
  for $p_1, p_2 \in (0,1), \alpha, \beta \in (0,1]$.
\end{lem}
\begin{proof}
  Let $$ g(x) = (1-x)^\gamma - (1 - \gamma x),$$
  for $x \in [0,1), \gamma \in (0,1]$.

  By taking the derivatives,
  $$ g'(x) = - \gamma (1 - x)^{\gamma - 1} + \gamma;$$
  $$ g''(x) = \gamma(\gamma - 1)(1 - x)^{\gamma - 2}.$$

  Since $g(0) = 0$, according to the mean value theorem,
  $$ g(x) = g'(\xi) x, $$
  where $ 0 < \xi \leq x < 1$. Substituting $g(x)$ and $g'(\xi)$,
  \begin{eqnarray}\label{eq:sub}
  (1 - x)^\gamma - (1 - \gamma x) &=& g'(\xi) x \nonumber \\
    1 - (1 - x)^\gamma &=& \gamma x - g'(\xi) x \nonumber \\
    1 - (1 - x)^\gamma &=& \gamma x (1 - \xi)^{\gamma - 1}.
  \end{eqnarray}

  Given $g''(x) < 0$ for $x,\gamma \in (0,1)$, $g(x)$ is strictly concave for $x, \gamma \in (0,1)$. For $0 < x_1 < x_2 < 1$,
  \begin{eqnarray*}
    g(x_2) &<& g(0) +  \frac{g(x_1) - g(0)}{x_1} x_2, \\
    g(x_2)/x_2 &<& g(x_1)/x_1.
  \end{eqnarray*}
  Given that $g(x_1) = g'(\xi_1) x_1$, $g(x_2) = g'(\xi_2) x_2$, $0 < \xi_1 < x_1$, $0 < \xi_2 < x_2$,
  and that $g'(x)$ is decreasing in $x$, we have $\xi_1 < \xi_2$. If $\gamma = 1$, then $g(x) = g'(x) = 0$ for $x \in [0,1)$. Clearly, there also exist $\xi_1 < \xi_2$ such that $g(x_1) = g'(\xi_1) x_1$, $g(x_2) = g'(\xi_2) x_2$.

  Applying Eq. (\ref{eq:sub}),
  \begin{eqnarray*}
    &~&[1 - (1 - p_1)^\alpha][1 - (1 - p_2)^\beta] \\
    &=& \alpha p_1 (1 - \eta_1)^{\alpha - 1} \beta p_2 (1 - \eta_2)^{\beta - 1} \\
    &=& \alpha \beta p_1 p_2 (1 - \eta_1)^{\alpha - 1} (1 - \eta_2)^{\beta - 1},
  \end{eqnarray*}
  for $0< \eta_1 < p_1$, $0 < \eta_2 < p_2$, and
  \begin{eqnarray*}
    1 - (1 - p_1 p_2)^{\alpha \beta}
    = \alpha \beta p_1 p_2 (1 - \eta_3)^{\alpha \beta - 1},
  \end{eqnarray*}
  for $0 < \eta_3 < p_1 p_2$. Moreover, $\eta_3 \leq \min(\eta_1, \eta_2)$, because $p_1 p_2 \leq \min(p_1,p_2)$.

  Given $0 < \alpha, \beta \leq 1$,
  \begin{eqnarray*}
    &~&(1 - \eta_1)^{\alpha - 1} (1 - \eta_2)^{\beta - 1} \\
    &\geq& (1 - \eta_3)^{\alpha - 1} (1 - \eta_3)^{\beta - 1} \\
    &=& (1-\eta_3)^{\alpha + \beta - 2} \\
    & \geq & (1 - \eta_3)^{\alpha \beta - 1},
  \end{eqnarray*}
  where the first inequality follows from the fact that $h_1(x) = (1 - x)^\alpha$ is decreasing in $x$ if $x \in (0,1)$ and $\alpha \in (0,1]$, and the last inequality follows from
  \begin{eqnarray*}
    \alpha(1 - \beta) &\leq& 1 - \beta, \\
    \alpha + \beta - 2 & \leq& \alpha \beta - 1,
  \end{eqnarray*}
  and $h_2(x) = (1 - \eta_3)^x$ is decreasing in $x$ for $\eta_3 \in (0,1)$.

  Therefore,
  \begin{equation*}
    1 - (1 - p_1 p_2)^{\alpha \beta} \leq [1 - (1 - p_1)^\alpha] [1 - (1 - p_2)^\beta].
  \end{equation*}
\end{proof}

\begin{lem} \label{th:inequality0}
  $$1 - p(v_i) \geq (1 - \tilde p(v_i))^{n_d^{n_s}},$$
  where $\tilde p(v_i)$ is defined before Lemma \ref{th:lower}.
\end{lem}
\begin{proof}
  For $n_s = 1$, every node has a single supply node. Let $u_j^i$ be the supply node of $v_i$. Since $1 - p(u_j^i) \geq (1 - \tilde p(u_j^i))^{n_d}$ for any supply node $u_j^i$, the result trivially holds.

  We next focus on the case where $n_s \geq 2$. Recall that $p(v_i) = \prod_{u_j^i}p(u_j^i)$ and $\tilde p(v_i) = \prod_{u_j^i}\tilde p(u_j^i)$ (\ie, a demand node fails if and only if all of its supply nodes fail), where $u_j^i$ are the supply nodes of $v_i$.

  Consider two supply nodes of $v_i$ and let $p(u_1^i)$ and $p(u_2^i)$ be their failure probabilities. Moreover, $\tilde p(u_1^i)$ and $\tilde p(u_2^i)$ satisfy $\tilde p(u_1^i) \geq 1 - (1 - p(u_1^i))^{1/n_d}$ and $\tilde p(u_2^i) \geq 1 - (1 - p(u_2^i))^{1/n_d}$.

  Then,
  \begin{align*}
    1-\tilde p(u_1^i) \tilde p(u_2^i) & \leq 1 - [(1 - (1 - p(u_1^i))^{1/n_d})\\
    & ~~~~~~~~~(1 - (1 - p(u_2^i))^{1/n_d})] \\
    & \leq (1 - p(u_1^i) p(u_2^i))^{1/n_d^2},
  \end{align*}
  where the last inequality follows from Eq. (\ref{eq:logprob}), by letting $p_1 = p(u_1^i), p_2 = p(u_2^i), \alpha, \beta = 1/n_d$, which we proved in Lemma \ref{th:inequality}.
  \begin{equation}\label{eq:logprob}
    1 - (1 - p_1 p_2)^{\alpha \beta} \leq [1 - (1 - p_1)^\alpha] [1 - (1 - p_2)^\beta],
  \end{equation}
  for $p_1, p_2 \in (0,1), \alpha, \beta \in (0,1]$.

  Consider the third supply node of $v_i$ which has failure probability $p(u_3^i)$. We have $\tilde p(u_3^i) \geq 1 - (1 - p(u_3^i))^{1/n_d}$. Moreover, notice that $\tilde p(u_1^i) \tilde p(u_2^i) \geq 1 - (1 - p(u_1^i) p(u_2^i))^{1/n_d^2}$. By letting $p_1 = p(u_1^i) p(u_2^i), p_2 = p(u_3^i), \alpha = 1/n_d^2, \beta = 1/n_d$ in Eq. \ref{eq:logprob}, we have
  $$ 1-\tilde p(u_1^i) \tilde p(u_2^i) \tilde p(u_3^i)\leq (1 - p(u_1^i) p(u_2^i) p(u_3^i))^{1/n_d^3}.$$

  By repeating the process until all the supply nodes of $v_i$ are considered, and let $n_s(v_i) \leq n_s$ denote the number of supply nodes of $v_i$, we have
  \begin{align*}
    1 - \tilde p(v_i) \leq & (1 - p(v_i))^{1/n_d^{n_s(v_i)}} \\
    \leq & (1 - p(v_i))^{1/n_d^{n_s}}.
  \end{align*}
\end{proof}
\typeout{}
\bibliographystyle{IEEEtran}
\bibliography{vertexcut.bib}
\end{document}